\documentclass[a4paper,11pt]{article}
\usepackage{amsmath,amssymb,amsfonts,amsthm,stmaryrd}
\usepackage{color}
\usepackage{geometry}
\usepackage{tikz}
\usepackage{graphicx}
\usepackage{url}
\usepackage{hyperref}
\hypersetup{
    colorlinks=true,
    citecolor=blue,
    linkcolor=red,
    filecolor=magenta,
    urlcolor=cyan,
}
\let\set\mathbb

\newtheorem{theorem}{Theorem}[section]

\newtheorem{lemma}[theorem]{Lemma}
\newtheorem{remark}[theorem]{Remark}
\newtheorem{definition}[theorem]{Definition}
\newtheorem{example}[theorem]{Example}
\newtheorem{proposition}[theorem]{Proposition}
\newtheorem{fact}[theorem]{Fact}

\def\redt{\operatorname{red}}
\def\im{\operatorname{im}}
\def\spanning{\operatorname{span}}

\newcommand{\qbinom}{\genfrac{[}{]}{0pt}{}}

\title{Order Bounds for Hypergeometric and $q$-Hypergeometric Creative Telescoping}
\author{Hui Huang\\
	{\small School of Mathematics and Statistics, Fuzhou University,}\\
	{\small Fuzhou 350108, China}\\
	{\tt \small huanghui@fzu.edu.cn}
	}
\date{}

\begin{document}
\maketitle

\begin{abstract}
Leveraging a general framework adapted from symbolic integration, a unified reduction-based 
algorithm for computing telescopers of minimal order for hypergeometric and $q$-hypergeometric terms 
has been recently developed. In this paper, we conduct a deeper exploration and put forth a new 
argument for the termination of the algorithm. 
This not only provides an independent proof of existence of telescopers, but also 
allows us to derive unified upper and lower 
bounds on the order of telescopers for hypergeometric terms and their $q$-analogues. 
Compared with known bounds in the literature, our bounds,
in the hypergeometric case, are exactly the same as the 
tight ones obtained in 2016; while in the $q$-hypergeometric case, 
no lower bounds were known before, and our upper bound is 
sometimes better and never worse than the known one.
\end{abstract}

\section{Introduction}\label{SEC:intro}
This paper concerns about hypergeometric terms and their $q$-analogues.
These are ubiquitous objects appearing in combinatorics. A sequence, say 
$f_{n,k}$ in  two discrete variables $n,k$, is called a {\em hypergeometric term} 
if the two shift quotients $f_{n+1,k}/f_{n,k}$ and $f_{n,k+1}/f_{n,k}$ are 
rational functions in $n$ and $k$, and it is called a {\em $q$-hypergeometric term} 
if the two shift quotients are rational functions in $q^n$ and~$q^k$.
Prototypical examples are respectively given by the binomial coefficient $f_{n,k} = {n\choose k}$
and the Gaussian binomial coefficient
\[
f_{n,k} = \qbinom{n}{k}_q=\begin{cases}
\frac{(q;q)_n}{(q;q)_k(q;q)_{n-k}} & \text{if}\ 0\leq k\leq n,\\[1ex]
0 & \text{otherwise},
\end{cases}
\]
where $(q;q)_k=\prod_{i=1}^k(1-q^i)$ is the $q$-Pochhammer symbol with $(q;q)_0=1$. 

The method of creative telescoping, pioneered by Zeilberger~\cite{Zeil1990a,Zeil1990b,Zeil1991},
currently serves as the primary technique for 
simplifying definite sums of hypergeometric terms and their $q$-analogues. 
Specializing to the hypergeometric case, the method takes a hypergeometric 
summand $f_{n,k}$ and finds a recurrence operator $L$ independent of $k$ 
and the $k$-shift operator $S_k$, and another hypergeometric term $g_{n,k}$ such that 
$L(f_{n,k}) = g_{n,k+1}-g_{n,k}$. More precisely, the operator $L$ has the form 
$L = c_0+c_1S_n+\cdots+c_\rho S_n^\rho$ for some $c_0,\dots,c_\rho$ only 
depending on~$n$,
where $S_n$ is the $n$-shift operator.
We call such an $L$ a {\em telescoper} for $f_{n,k}$ and $g_{n,k}$ 
a {\em certificate} for~$L$. Telescopers and certificates allow us to produce recurrence 
equations solved by definite sums of interest such as $\sum_{k=0}^nf_{n,k}$, yielding
useful information about the sums; see \cite{PWZ1996} for further details.

During the past 35 years, the method of creative telescoping has been 
generalized and refined in various ways.
The latest trend in the development of creative telescoping is the so-called 
{\em reduction-based approach} originated from \cite{BCCL2010}. 
One important feature of this approach is that it is 
flexible to construct a telescoper for a given term with or without construction of a certificate.
This is particularly useful in a typical situation where the certificate is not needed and it
is much larger (and thus computationally more expensive) than the telescoper.
An excellent exposition of this approach can be found in~\cite{Chen2019}.

As indicated by the name, reductions play a fundamental role in the reduction-based approach.
Let $\set K$ be a field of characteristic zero. By \cite[Definition~5.67]{Kaue2023}, 
a {\em reduction} $\redt_k$, formulated for the shift case, is a $\set K(n)$-linear map from a
domain under consideration, say~$\set D$, containing $\set K(n,k)$, to itself with the property that 
for all $f\in \set D$ there exists a $g\in\set D$ such that $f-\redt_k(f)=S_k(g)-g$. In other words, 
the difference $f-\redt_k(f)$ is a summable term. We call $\redt_k(f)$ a {\em remainder} of $f$ 
with respect to the reduction $\redt_k$. The basic idea of the reduction-based approach
is then as follows. For a given $f\in\set D$, we continually compute 
$\redt_k(f),\redt_k(S_n(f)),\redt_k(S_n^2(f)),\dots$ 
until we find a nontrivial linear dependency over~$\set K(n)$. Once we have such a dependency,
say $c_0\redt_k(f)+\cdots+c_\rho\redt_k(S_n^\rho(f))=0$ for some $c_0,\dots,c_\rho\in \set K(n)$,
not all zero, then the operator $c_0+\cdots+c_\rho S_n^\rho$ is a telescoper for~$f$.

There are usually two common ways to guarantee the termination of the above process.
One is to show that the $\set K(n)$-vector space spanned by 
$\redt_k(f),\redt_k(S_n(f)),\redt_k(S_n^2(f)),\dots$ 
has a finite dimension, or equivalently, the reduction $\redt_k$ is {\em confined} 
(cf.\ \cite[Definition~5.67]{Kaue2023}). Then, as soon as $\rho$ exceeds this dimension, 
we can be sure that a nontrivial linear dependency (and thus a telescoper) will be found. 
The advantage 
of this way is that it not only provides an independent proof of existence of telescopers but 
also yields bounds on their orders. Order bounds for telescopers will be useful in analyzing
the complexity and improving the efficiency of creative telescoping algorithms.
The other way to ensure the termination is to 
show that for every summable term $f$ we have $\redt_k(f) = 0$, or equivalently, the reduction 
$\redt_k$ is {\em complete} (cf.\ \cite[Definition~5.67]{Kaue2023}). This makes sure that we do
not miss any telescoper and the smallest possible one will be found provided that telescopers
are known to exist. This way requires a priori knowledge of existence of telescopers and does
not provide any information on the order. We remark that existence problems of telescopers 
have already been solved for hypergeometric terms~\cite{Abra2003}, for their $q$-analogues
\cite{CHM2005}, for trivariate rational functions~\cite{CHLW2016,CDZ2019,CDWZ2021},
and more recently for P-recursive sequences~\cite{Du2025}.

Recently, Chen et al.\ \cite{CDGHL2025} presented a unified reduction, and 
a subsequent creative telescoping algorithm, for both hypergeometric and 
$q$-hypergeometric terms using a 
general framework adapted from symbolic integration. They showed that this reduction 
is complete. The goal of the present paper is to continue their theory and further prove 
that the reduction is also confined. We will reformulate the theory in certain way so as 
to obtain a finite-dimensional $\set K(n)$-vector space produced by remainders. This
enables us to derive unified upper and lower bounds on the order of telescopers 
for hypergeometric and $q$-hypergeometric terms, 
providing a second, alternate way, in addition to the existence criterion, 
to ensure the termination of creative telescoping algorithms.
Compared with known bounds in the literature, our bounds,
in the hypergeometric case, are exactly the same as those tight ones
provided in~\cite{Huan2016}. In the $q$-hypergeometric case, 
no lower bounds were known before, and our upper bound is 
at least as good as the known one in~\cite{MoZe2005} especially 
for ``generic input" (as the latter is already generically sharp); however, 
there are examples in which our bound is better than the known bound.

The remainder of the paper proceeds as follows. The unified theory developed in~\cite{CDGHL2025}
will be recalled in the next section. In Section~\ref{SEC:rem}, we present an inherent 
relation between remainders, which enables us to obtain in Section~\ref{SEC:ord} unified 
upper and lower bounds on the order of telescopers for hypergeometric and $q$-hypergeometric terms.
The paper ends with a comparison between our bounds with the known ones in the 
literature.

\section{A unified theory}\label{SEC:prelim}
Throughout the paper, we let $\set K$ be a field of characteristic zero with $\set F = \set K(x)$
and $\set F(y)=\set K(x,y)$ being the field of rational functions in $x$ and $y$ over $\set K$. 
Let $\sigma_x$ and $\sigma_y$ be both either the usual shift operators with respect to $x$ 
and $y$ respectively defined by
\[
\sigma_x(f(x,y)) = f(x+1,y)\quad\text{and}\quad
\sigma_y(f(x,y)) = f(x,y+1),
\]
or the $q$-shift operators with respect to $x$ and $y$ respectively
defined by
\[
\sigma_x(f(x,y)) = f(qx,y)\quad\text{and}\quad
\sigma_y(f(x,y)) = f(x,qy)
\]
for any $f\in \set K(x,y)$, where $q\in \set K$ is neither zero nor a root of unity.  
We will refer to the former case as the usual shift case, and the latter one as the $q$-shift case. 
The pair $(\set K(x,y),\{\sigma_x,\sigma_y\})$ forms a partial ($q$-)difference field. 
Let $R$ be a {\em partial ($q$-)difference ring extension} of $(\set K(x,y),\{\sigma_x,\sigma_y\})$, 
that is, $R$ is a ring containing $\set K(x,y)$ together with two distinguished endomorphisms 
$\sigma_x$ and $\sigma_y$ from $R$ to itself, whose respective restrictions to $\set K(x,y)$ 
agree with the two automorphisms defined earlier. An element $c\in R$ is called a {\em constant} 
if it is invariant under the applications of $\sigma_x$ and~$\sigma_y$.  It is readily seen that 
all constants in $R$ form a subring of $R$.
\begin{definition}
An invertible element $T$ of~$R$ is called a {\em $\sigma_y$-hypergeometric term}
if there exists $f\in \set F(y)$ such that $\sigma_y(T) = fT$.  We call $f$ the 
$\sigma_y$-quotient of~$T$. An invertible element $T$ of~$R$ is called 
a {\em $(\sigma_x,\sigma_y)$-hypergeometric term} if it is both $\sigma_x$- 
and $\sigma_y$-hypergeometric. 
\end{definition}
For a given $\sigma_y$-hypergeometric term $T$, an important question is 
to decide whether there exists another $\sigma_y$-hypergeometric term $G$ 
such that $T = \Delta_y(G)$, 
where $\Delta_y$ denotes the difference of~$\sigma_y$ and the identity map of~$R$.
If such a $G$ exists, we say that $T$ is {\em $\sigma_y$-summable}. We refer to this 
question as the {\em $\sigma_y$-summability problem}. A key idea on solving the 
$\sigma_y$-summability problem of a $\sigma_y$-hypergeometric term 
$T$ is to write it into a multiplicative decomposition $T = fH$, where $f\in \set F(y)$ 
and $H$ is a $\sigma_y$-hypergeometric term enjoying certain nice properties 
(cf.\ \cite{AbPe2002b,CHM2005}), and then reduce the problem to find a rational function 
$g\in \set F(y)$ such that
\[
f H= \Delta_y(gH), \ \text{or equivalently,}\ f = \Delta_K(g),
\]
where $K = \sigma_y(H)/H$ and $\Delta_K=K\sigma_y-1$ is the linear map from $\set F(y)$ 
to itself mapping $f\in\set F(y)$ to $K\sigma_y(f)-f$. This brings our attentions from the
$\sigma_y$-hypergeometric terms to rational functions, and motivates us to 
recall \cite[\S 2.1]{CDGHL2025} the notion of normal and special polynomials.
\begin{definition}\label{DEF:types}
A polynomial $p\in \set F[y]$ is said to be {\em $\sigma_y$-normal} if $\gcd(p,\sigma_y^\ell(p)) = 1$ 
for any nonzero integer~$\ell$, and {\em $\sigma_y$-special} if $p\mid \sigma_y^\ell(p)$ for some 
nonzero integer $\ell$.
\end{definition}
Note that $\sigma_y$-normal polynomials coincide with $\sigma_y$-free polynomials in the literature.
A polynomial is not necessarily $\sigma_y$-normal or $\sigma_y$-special, but an irreducible 
polynomial in~$\set F[y]$ must be either $\sigma_y$-normal or $\sigma_y$-special. All elements
in~$\set F$ are $\sigma_y$-special, and a polynomial in~$\set F[y]$ is both $\sigma_y$-normal
and $\sigma_y$-special if and only if it is in $\set F\setminus\{0\}$. 

Recall that two polynomials $a,b\in \set F[y]$ are {\em $\sigma_y$-coprime} if 
$\gcd(a,\sigma_y^\ell(b)) = 1$ for any nonzero integer~$\ell$. Note that two 
$\sigma_y$-coprime polynomials need not to be coprime. The following describes
nice multiplicative properties satisfied by $\sigma_y$-normal and $\sigma_y$-special
polynomials.
\begin{proposition}[{\cite[Proposition~2.4]{CDGHL2025}}]
\leavevmode\null
\label{PROP:props}
\begin{itemize}
\item[(i)] Any finite product of $\sigma_y$-normal and pairwise $\sigma_y$-coprime 
polynomials in~$\set F[y]$ is $\sigma_y$-normal. Any factor of a $\sigma_y$-normal polynomial 
in $\set F[y]$ is $\sigma_y$-normal.
\item[(ii)] Any finite product of $\sigma_y$-special polynomials in $\set F[y]$ is $\sigma_y$-special.
Any factor of a nonzero $\sigma_y$-special polynomial in $\set F[y]$ is $\sigma_y$-special.
\end{itemize}
\end{proposition}
It is known that a major difference from the $q$-shift case to the usual
one lies in the appearance of nontrivial $\sigma_y$-special polynomials.
\begin{lemma}[{\cite[Lemma~2.7]{CDGHL2025}}]\label{LEM:special}
Let $p$ be a polynomial in $\set F[y]$. 
\begin{itemize}
\item[(i)] In the usual shift case, $p$ is $\sigma_y$-special if and only if $p\in \set F$.
\item[(ii)] In the $q$-shift case, $p$ is $\sigma_y$-special 
if and only if $p/y^k\in \set F$ for some $k\in \set N$.
\end{itemize}
\end{lemma}
Using the theory of $\sigma_y$-normal and $\sigma_y$-special polynomials,
the authors of \cite{CDGHL2025} presented a unified reduction for 
$\sigma_y$-hypergeometric terms, addressing the $\sigma_y$-summability
problem without solving any auxiliary difference equations. To describe it
concisely, we need to recall some terminology.

Based on~\cite{AbPe2002b,CHM2005}, a nonzero rational function in $\set F(y)$ with numerator~$u$ 
and denominator~$v$ is said to be {\em $\sigma_y$-reduced} if $\gcd(u,\sigma_y^\ell(v)) = 1$
for any integer~$\ell$.
For a nonzero rational function $f\in \set F(y)$, there exist
two nonzero rational functions $K,S$ in $\set F(y)$ with $K$ being $\sigma_y$-reduced such that
\[
f  = K\frac{\sigma_y(S)}S.
\]
Such a pair $(K,S)$ will be called a {\em rational normal form} (or an {\em RNF} for short) of~$f$.
Moreover, we call $K$ a {\em kernel} and $S$ a corresponding {\em shell} of~$f$. 
These quantities can be constructed by gcd-calculations (cf.\ \cite{AbPe2002b,CHM2005}).
\begin{definition}
Let $K\in \set F(y)$ with numerator~$u$ and denominator~$v$. 
A nonzero polynomial $p\in \set F[y]$ is said to be {\em strongly coprime} with $K$ if 
$\gcd(u,\sigma_y^\ell(p)) = \gcd(v,\sigma_y^{-\ell}(p)) = 1$ for all $\ell\in \set N$.
\end{definition}
The following notion is crucial in developing the unified reduction. 
\begin{definition}\label{DEF:standard}
A rational function in~$\set F(y)$ with numerator $u$ and denominator $v$ is said to be {\em $\sigma_y$-standard}
if it is $\sigma_y$-reduced, and in the $q$-shift case, $u(0)q^\ell-v(0)\neq 0$ for any negative integer~$\ell$.
\end{definition}
By \cite[Proposition~3.16]{CDGHL2025}, every nonzero rational
function in $\set F(y)$ has a $\sigma_y$-standard kernel which
can be easily computed.

For a nonzero polynomial $p\in \set F[y]$, its degree in~$y$ (or
$y$-degree) is denoted by~$\deg_y(p)$. We will follow the convention that
$\deg_y(0) = -\infty$.  A rational function in $\set F(y)$ is said to
be {\em proper} if the $y$-degree of its numerator is less than that
of its denominator.
Let $K\in \set F(y)$ be $\sigma_y$-standard with numerator~$u$ and denominator~$v$. 
We define an $\set F$-linear map $\phi_K$
from $\set F[y]$ to itself by sending $p$ to $u\sigma_y(p)-vp$ for all $p\in \set F[y]$,
and call it the {\em map for polynomial reduction with respect to~$K$}.
The image of~$\phi_K$, denoted by $\im(\phi_K)$, is an $\set F$-linear subspace of~$\set F[y]$.
Let
\[
\im(\phi_K)^\top = \spanning_{\set F}\big\{y^d\mid d\in \set N \ \text{and}\ d \neq \deg_y(p)\ 
\text{for all}\ p \in \im(\phi_K)\big\}.
\]
Then $\set F[y] = \im(\phi_K) \oplus\im(\phi_K)^\top$, and thus we call $\im(\phi_K)^\top$ the 
{\em standard complement of~$\im(\phi_K)$}.

\begin{definition}\label{DEF:remainder}
Let $K\in \set F(y)$ be $\sigma_y$-standard with denominator~$v$, and let $f\in\set F(y)$. 
Another rational function
$r$ in $\set F(y)$ is called a {\em $\sigma_y$-remainder} of $f$ with respect to~$K$ if 
$f-r\in \im(\Delta_K)$ and $r$ can be written in the form
\begin{equation}\label{EQ:remainder}
r = h + \frac{p}v,
\end{equation}
where $h\in \set F(y)$ is proper with denominator being $\sigma_y$-normal and strongly
coprime with~$K$, and $p\in \im(\phi_K)^\top$. For brevity, we just say that $r$
is a $\sigma_y$-remainder with respect to~$K$ if $f$ is clear from the context.
In addition, we call the denominator of~$h$ the {\em significant denominator} of~$r$.
\end{definition}
The unified reduction developed in~\cite{CDGHL2025} can be stated as follows.
\begin{theorem}[{\cite[Theorem~3.18]{CDGHL2025}}]\label{THM:shellred}
For a $\sigma_y$-hypergeometric term $T$ whose $\sigma_y$-quotient
has a $\sigma_y$-standard kernel $K$ and a corresponding shell~$S$, 
there exists an algorithm computing a rational function $g\in\set F(y)$ 
and a $\sigma_y$-remainder $r$ with respect to~$K$ such that 
\begin{equation}\label{EQ:hyperred}
T = \Delta_y(gH) + rH,
\end{equation}
where $H = T/S$. Moreover, $T$ is $\sigma_y$-summable if and only if $r= 0$.
\end{theorem}

In the context of creative telescoping, we will also need to consider the variable~$x$.
Let $\set F[S_x]$ be the ring of linear recurrence operators in~$x$ over
$\set F$, in which the commutation rule is that $S_xf =
\sigma_x(f)S_x$ for all $f\in \set F$.  The application of an operator
$L = \sum_{i=0}^\rho \ell_iS_x^i\in \set F[S_x]$ to a
$(\sigma_x,\sigma_y)$-hypergeometric term $T$ is defined as
\[
L(T) = \sum_{i=0}^\rho \ell_i\sigma_x^i(T).
\]
\begin{definition}
Let $T$ be a $(\sigma_x,\sigma_y)$-hypergeometric term.  A nonzero
linear recurrence operator $L\in \set F[S_x]$ is called a {\em telescoper} for $T$ if there exists a
$(\sigma_x,\sigma_y)$-hypergeometric term $G$ such that
\[
L(T) = \Delta_y(G).
\] 
We call $G$ a corresponding {\em certificate} of~$L$. The {\em order}
of a telescoper for~$T$ is defined to be its degree in $S_x$.
\end{definition}

Based on the unified reduction, an efficient creative telescoping algorithm for
constructing a telescoper for a $(\sigma_x,\sigma_y)$-hypergeometric
term was further proposed in~\cite{CDGHL2025}. We briefly summarize its 
main idea below.

Let $T$ be a $(\sigma_x,\sigma_y)$-hypergeometric term 
whose $\sigma_y$-quotient has a $\sigma_y$-standard 
kernel $K$ and a corresponding shell~$S$. Applying
Theorem~\ref{THM:shellred} to $T$ yields
\[
T = \Delta_y(g_0H) + r_0H,
\]
where $H=T/S$, $g_0\in \set F(y)$ and $r_0$ is a $\sigma_y$-remainder 
with respect to~$K$. If the significant denominator of $r_0$ is not
integer-linear (see Definition~\ref{DEF:intlinear}), then by existence 
criterion (cf.\ \cite{Abra2003,CHM2005}), $T$ admits no telescopers, 
so the algorithm terminates. Otherwise, telescopers for~$T$ must exist. 
One proceeds by applying 
Theorem~\ref{THM:shellred}, along with a manipulation of resulting
$\sigma_y$-remainders by \cite[Theorem~4.1]{CDGHL2025},
to $\sigma_x^i(T)$ 
incrementally with $i = 1,2,\dots$ to get $g_i\in \set F(y)$ and a 
$\sigma_y$-remainder~$r_i$ with respect to $K$ such that
\begin{equation}\label{EQ:rct}
\sigma_x^i(T) = \Delta_y(g_iH) + r_iH \quad \text{and}\quad
\sum_{j=0}^ic_jr_j\ \text{is a $\sigma_y$-remainder for all $c_j\in \set F$,}
\end{equation}
and finding a nontrivial $\set F$-linear dependency among the $r_i$.
Any such an $\set F$-linear dependency will trigger the algorithm to
return a desired telescoper for~$T$.

The termination of the above algorithm is guaranteed by existence
criterion, see Theorem~5.7 in~\cite{CDGHL2025} for more details. 
Instead of using the criterion, in the rest of the paper we 
are going to prove the termination by showing that the
$\sigma_y$-remainders $r_i$ produced in \eqref{EQ:rct} form a 
finite-dimensional vector space over~$\set F$, and along the way
obtain upper and lower bounds on the order of telescopers,
generalizing the results in \cite{Huan2016} to also cover the 
$q$-shift case.

\section{Relations between remainders}\label{SEC:rem}
We first discuss in 
this section an inherent relation between two $\sigma_y$-remainders relevant to 
different $x$-shifts of a $(\sigma_x,\sigma_y)$-hypergeometric term.
The following definition plays a crucial role in describing such a relation.
\begin{definition}\label{DEF:srelated}
Two $\sigma_y$-normal polynomials $a,b\in \set F[y]$ are called {\em $\sigma_y$-related} if 
for any irreducible factor $f\in\set F[y]$ of~$a$ over~$\set F$ with multiplicity~$k$, there exists 
an integer $\ell$ such that $\sigma_y^\ell(f)$ is a factor of~$b$ with the same multiplicity~$k$, 
and vice versa.
\end{definition}
Note that the integer $\ell$ in the above definition is unique since both $a$ and $b$ 
are $\sigma_y$-normal. Clearly, the $\sigma_y$-relatedness is an equivalence relation, and
two $\sigma_y$-related polynomials in $\set F[y]$ have the same $y$-degree.

The proof of Lemma~3.2 in \cite{CDGHL2025} actually contains a more generalized
conclusion given below, which will be useful in exploring the desired relation.
\begin{lemma}\label{LEM:minimal}
Let $K\in \set F(y)$ with denominator~$v$, and let $h\in \set F(y)$ be
a rational function whose denominator $d$ is $\sigma_y$-normal and
strongly coprime with~$K$. Assume that there are $\tilde h\in \set F(y)$
and $p\in \set F[y]$ such that
\[
h - \tilde h + \frac{p}v \in \im(\Delta_K).
\]
Then for any irreducible factor $a\in \set F[y]$ of~$d$ over~$\set F$ with
multiplicity~$k$, there exists an integer~$\ell$ such that
$\sigma_y^\ell(a)^k$ divides the denominator of~$\tilde h$.
\end{lemma}
For a nonzero rational function in $\set F(y)$, the following proposition reveals the relation
between $\sigma_y$-remainders of its two shells with respect to the same kernel.
\begin{proposition}\label{PROP:onekernel}
Let $(K,S)$ and $(K,\tilde S)$ be two RNFs of a nonzero rational function $f\in\set F(y)$,
with $K$ being $\sigma_y$-standard.
Let $r$ and $\tilde r$ be $\sigma_y$-remainders of~$S$ and $\tilde S$ with respect to~$K$, respectively. 
Then the significant denominators of $r$ and $\tilde r$ are $\sigma_y$-related.
\end{proposition}
\begin{proof}
Since $(K,S)$ and $(K,\tilde S)$ are both RNFs of~$f$, we have
$f=K\sigma_y(S)/S=K\sigma_y(\tilde S)/\tilde S$,
and then $\sigma_y(\tilde S/S) = \tilde S/S$. Notice that in the $q$-shift case, $q$ is not a root of unity.
So it follows from~\cite[Lemma~2.6]{CDGHL2025} that in both usual shift and $q$-shift cases,
$\tilde S/S\in \set F$, that is, $\tilde S = c S$ for some nonzero element $c\in \set F$.
Since $r$ and $\tilde r$ are $\sigma_y$-remainders of~$S$ and $\tilde S$
with respect to~$K$, respectively, we have $S-r \in \im(\Delta_K)$ and
$\tilde S - \tilde r \in \im(\Delta_K)$. Thus $c\,r-\tilde r\in \im(\Delta_K)$.

Let $d$ and $\tilde d$ in $\set F[y]$ be the significant denominators of~$r$ 
and~$\tilde r$, respectively, and write
\[
r = h + \frac{p}v\quad\text{and}\quad \tilde r= \tilde h + \frac{\tilde p}v,
\]
where $h,\tilde h\in \set F(y)$ are proper with denominators $d,\tilde d$, respectively,
$p,\tilde p\in \im(\phi_K)^\top$,
and $v$ is the denominator of~$K$.
It follows that
\[
c\,h-\tilde h + \frac{c\,p-\tilde p}v \in \im(\Delta_K).
\]
Note that $c\,h$ and $\tilde h$ are both proper rational functions in~$\set F(y)$
whose denominators $d$ and $\tilde d$ 
are $\sigma_y$-normal and strongly coprime with~$K$.
If $d\in \set F$, then $h = 0$ and thus $-\tilde h+(c\,p-\tilde p)/v
\in \im(\Delta_K)$, yielding $\tilde h = 0$ by \cite[Theorem~3.5]{CDGHL2025},
so $\tilde d\in \set F$ and  the assertion clearly holds.
Assume that $d\notin \set F$ and let $a\in\set F[y]$ be an irreducible factor of~$d$ over~$\set F$ 
with multiplicity~$k$.
It then follows from Lemma~\ref{LEM:minimal} that there exists an integer $\ell$ such that 
$\sigma_y^\ell(a)^k$ divides~$\tilde d$. This implies that $\sigma_y^\ell(a)$ is an irreducible
factor of~$\tilde d$ over~$\set F$ whose multiplicity, say~$k'$, is at least~$k$. 
By applying Lemma~\ref{LEM:minimal} to~$\tilde h$, we see that 
$\sigma^{\ell+i}(a)^{k'}$ divides $d$ for some $i\in \set Z$. 
Since $d$ is $\sigma_y$-normal and $a$ is a factor of $d$ with multiplicity~$k$, 
we have $i = -\ell$ and $k'\leq k$, which yields that $k'=k$.
Thus $\sigma_y^\ell(a)$ is a factor of~$\tilde d$ with multiplicity~$k$. 
Now by switching the roles of $c\,h$ and $\tilde h$ (and thus $d$ and $\tilde d$) in the above arguments,
we conclude that for any irreducible factor $\tilde a\in\set F[y]$ of~$\tilde d$ over~$\set F$
with multiplicity~$\tilde k$, there exists an integer $\tilde \ell$ such that $\sigma_y^{\tilde \ell}(\tilde a)$
is a factor of~$d$ with the same multiplicity~$\tilde k$. 
Therefore, by definition, $d$ and $\tilde d$ are $\sigma_y$-related.
\end{proof}

In order to study the case of different kernels, we recall some notions 
from~\cite{CDGHL2025}. 

A polynomial $p$ in $\set F[y]$ is said to be 
{\em $\sigma_y$-monic} if it is monic with respect to~$y$ in the usual shift case
or $q$-monic, namely $p\mid_{y=0}= 1$, in the $q$-shift case. A rational function in 
$\set F(y)$ is {\em $\sigma_y$-monic} if both its numerator and denominator 
are $\sigma_y$-monic. By a factor of a rational function in~$\set F(y)$, we mean 
a factor of either its numerator or its denominator. Let $f\in \set F(y)$ be 
$\sigma_y$-monic. Lemma~\ref{LEM:special} then tells us that all irreducible 
factors of~$f$ are $\sigma_y$-normal. Moreover, $\sigma_y^\ell(f)$ for all 
$\ell\in \set Z$ is again $\sigma_y$-monic. 
Let $p$ be a nonzero polynomial in~$\set F[y]$ and 
let $\alpha = \sum_{i=m}^nk_i\sigma_y^i$ be a Laurent polynomial in 
$\set Z[\sigma_y,\sigma_y^{-1}]$. We define
\[
p^\alpha = \prod_{i=m}^n\sigma_y^i(p)^{k_i}.
\]
Clearly, $p^\alpha$ is a polynomial if and only if $\alpha$ belongs to 
$\set N[\sigma_y,\sigma_y^{-1}]$. Recall that two polynomials $a,b$ 
in $\set F[y]$ are {\em associates} if $a = c\,b$ for $c\in \set F$.
According to \cite[Definition~11]{Karr1981} and \cite[Definition~1]{AbPe2002b}, 
two polynomials $a,b\in \set F[y]$ are said to be {\em $\sigma_y$-equivalent} if 
$a$ is an associate of~$\sigma_y^\ell(b)$ for some $\ell\in \set Z$. Evidently, this 
gives an equivalence relation, and the $\sigma_y$-equivalence of two polynomials
can be easily recognized by comparing coefficients. 

Let $f$ be a rational function in~$\set F(y)$. 
Separating $\sigma_y$-normal and $\sigma_y$-special irreducible factors 
of~$f$ over~$\set F$, we can decompose it as
\begin{equation}\label{EQ:splitfac}
f = f_sf_n,
\end{equation}
where $f_s,f_n\in \set F(y)$, the numerator and denominator of~$f_s$
are $\sigma_y$-special, and every irreducible factor of~$f_n$ over~$\set F$ is 
$\sigma_y$-normal. We will refer to \eqref{EQ:splitfac} as a 
{\em $\sigma_y$-splitting factorization} of~$f$. Such a factorization becomes
unique if we further require that $f_n$ be $\sigma_y$-monic. Grouping together 
$\sigma_y$-equivalent irreducible factors of~$f_n$ over~$\set F$, we further 
write $f$ in the form
\begin{equation}\label{EQ:sigmafac}
f = f_s p_1^{\alpha_1}\cdots\, p_m^{\alpha_m},
\end{equation}
where $m\in \set N, \alpha_i\in \set Z[\sigma_y,\sigma_y^{-1}]\setminus\{0\}$,
each $p_i\in \set F[y]$ is $\sigma_y$-monic and irreducible over~$\set F$,
and the~$p_i$ are pairwise $\sigma_y$-inequivalent. We call \eqref{EQ:sigmafac}
a {\em $\sigma_y$-factorization} of~$f$. 
\begin{lemma}\label{LEM:redprop}
Let $r\in \set F(y)$ with a $\sigma_y$-splitting factorization $r = r_s r_n$.
Assume that $r_n$ is $\sigma_y$-monic, and $r = \sigma_y(f)/f$ for some 
nonzero rational function $f\in \set F(y)$. Write $r_n = a/d$ with $a,d\in\set F[y]$ 
and $\gcd(a,d) = 1$. Then 
\begin{itemize}
\item[(i)] $r_s$ is equal to one in the usual shift case, or it is a power of $q$
in the $q$-shift case.
\item[(ii)] there exists a one-to-one correspondence $\varphi$ between the 
multi-sets of $\sigma_y$-monic irreducible factors of~$a$ and $d$ over~$\set F$ such that 
$p$ and $\varphi(p)$ are $\sigma_y$-equivalent for all $p\in\set F[y]$ dividing~$a$.
\end{itemize}
\end{lemma}
\begin{proof}
Assume that $f$ admits a $\sigma_y$-factorization of the form~\eqref{EQ:sigmafac}. 
Since $r = \sigma_y(f)/f$, we get
\[
r_sr_n = \frac{\sigma_y(f_s)}{f_s}\prod_{i=1}^mp_i^{\beta_i},
\]
where $\beta_i = \sigma_y\alpha_i - \alpha_i\neq 0$ for all $i=1,\dots,m$.
Notice that the numerator and denominator of~$f_s$ are $\sigma_y$-special
and all the~$p_i$ are $\sigma_y$-monic and irreducible over~$\set F$.
Since $r_n$ is $\sigma_y$-monic, we conclude from the uniqueness of
$\sigma_y$-splitting factorizations that 
\[
r_s = \frac{\sigma_y(f_s)}{f_s}\quad\text{and}\quad 
r_n = \prod_{i=1}^mp_i^{\beta_i}.
\]
Part (i) then immediately follows by Lemma~\ref{LEM:special}.
In order to verify part (ii), it suffices to show that the total sum of all 
coefficients of each $\beta_i$ with respect to~$\sigma_y$ is zero, 
which is in turn evident by the definition of the~$\beta_i$.
\end{proof}

The following explores the connection between two RNFs of a nonzero 
rational function in~$\set F(y)$. We remark that the usual shift version 
was already given in \cite[Theorem~2]{AbPe2002b}.
\begin{lemma}\label{LEM:twornfs}
Let $(K,S)$ and $(\tilde K,\tilde S)$ be two RNFs of a nonzero rational function $f\in\set F(y)$. Let
\[
K = K_sK_n\quad\text{and}\quad \tilde K = \tilde K_s\tilde K_n
\]
be $\sigma_y$-splitting factorizations of $K$ and~$\tilde K$, respectively, with $K_n,\tilde K_n$ 
being $\sigma_y$-monic. Then
\begin{itemize}
\item[(i)] $K_s/\tilde K_s$ is equal to one in the usual shift case, or it is a power of~$q$ in the $q$-shift case.
\item[(ii)] there exists a one-to-one correspondence $\varphi$ between the multi-sets of 
$\sigma_y$-monic irreducible factors of the denominators of~$K_n$ and $\tilde K_n$ over~$\set F$ such that 
$a$ and $\varphi(a)$ are $\sigma_y$-equivalent for all $a\in\set F[y]$ dividing the denominator of~$K_n$.
\item[(iii)] there exists a one-to-one correspondence $\psi$ between the multi-sets of 
$\sigma_y$-monic irreducible factors of the numerators of~$K_n$ and $\tilde K_n$ over $\set F$ such that 
$b$ and $\psi(b)$ are $\sigma_y$-equivalent for all $b\in\set F[y]$ dividing the numerator of~$K_n$.
\end{itemize}
\end{lemma}
\begin{proof}
Since $(K,S)$ and $(\tilde K,\tilde S)$ are both RNFs of~$f$,
we have
\[
f=K\frac{\sigma_y(S)}S = \tilde K\frac{\sigma_y(\tilde S)}{\tilde S},
\quad\text{that is,}\quad 
\frac{K}{\tilde K} = \frac{\sigma_y(\tilde S/S)}{\tilde S/S}.
\]
Notice that $K/\tilde K$ admits a $\sigma_y$-splitting factorization
$K/\tilde K = r_s r_n$,
where $r_s = K_s/\tilde K_s, r_n = K_n/\tilde K_n$ and $r_n$ is $\sigma_y$-monic.
The assertions follow by Lemma~\ref{LEM:redprop} and the observation that
both $K_n$ and $\tilde K_n$ are $\sigma_y$-reduced.
\end{proof}
The next two lemmas establish several transformations of RNFs which
preserve the significant denominators of associated $\sigma_y$-remainders.
Their proofs will take use of the following direct consequence
of the polynomial reduction developed in \cite[\S 3.3]{CDGHL2025}.
\begin{fact}\label{FAC:polyred}
Let $K\in \set F(y)$ be $\sigma_y$-standard with denominator~$v$. Then
for every polynomial $a\in \set F[y]$, there exist $g\in \set F(y)$ and
$p\in \im(\phi_K)^\top$ such that $a = \Delta_K(g) + p/v$. 
\end{fact}
\begin{lemma}\label{LEM:specialfac}
Assume that $\sigma_y$ is the $q$-shift operator. 
Let $(K,S)$ be an RNF of a nonzero rational function $f\in\set F(y)$ 
with $K$ being $\sigma_y$-standard, and 
let $r$ be a $\sigma_y$-remainder of $S$ with respect to~$K$.
Then for any $m\in \set N$, the pair
\[
(\tilde K,\tilde S) = (q^{-m}K, y^mS)
\]
is an RNF of~$f$ and $\tilde K$ is $\sigma_y$-standard. Moreover, there exists
a $\sigma_y$-remainder of~$\tilde S$ with respect to~$\tilde K$ which has the same 
significant denominator as~$r$.
\end{lemma}
\begin{proof}
Let $m\in \set N$ and set $(\tilde K,\tilde S) = (q^{-m}K, y^mS)$.
Since $K$ is $\sigma_y$-standard, so is~$\tilde K$.
The first assertion then follows by observing that
\[
f = K\frac{\sigma_y(S)}S = q^{-m} K\frac{\sigma_y(y^mS)}{y^mS}
= \tilde K\frac{\sigma_y(\tilde S)}{\tilde S}.
\]
Since $r$ is a $\sigma_y$-remainder of~$S$ with respect to~$K$, we have $S-r \in \im(\Delta_K)$,
and thus $S-r =\Delta_K(g)$ for some $g\in \set F(y)$. It follows that
$
\tilde S - y^m r
= y^m\Delta_K(g) 
= \Delta_{\tilde K}(y^mg)\in \im(\Delta_{\tilde K}).
$
It thus suffices to show that $y^m r$ has a $\sigma_y$-remainder with respect to $\tilde K$ whose
significant denominator is equal to that of~$r$. Write $r$ in the form $r = h+p/v$, where $h\in \set F(y)$ 
is proper whose denominator~$d$ is $\sigma_y$-normal and strongly coprime with~$K$, 
$p\in \im(\phi_K)^\top$,
and $v$ is the denominator of~$K$.
Then 
$
y^m r = y^m h + y^mp/v.
$
Since $d$ is $\sigma_y$-normal, it follows from Lemma~\ref{LEM:special}~(ii) that 
$y\nmid d$ and thus $d$ is again the 
denominator of~$y^mh$.
By division with remainder and a subsequent application of Fact~\ref{FAC:polyred}
on the quotient part, 
we can find $\tilde g\in \set F[y]$, $\tilde h\in \set F(y)$ and $\tilde p\in \im(\phi_{\tilde K})^\top$
such that
\[
y^m r = \Delta_{\tilde K}(\tilde g) + \tilde h+ \frac{\tilde p}{v},
\]
and $\tilde h$ is proper with denominator~$d$. 
Therefore, $\tilde h+\tilde p/v$ is a $\sigma_y$-remainder of~$y^mr$ (and thus $\tilde S$) 
with respect to~$\tilde K$ which has the same significant denominator as~$r$.
\end{proof}
\begin{lemma}\label{LEM:normalfac}
Let $(K,S)$ be an RNF of a nonzero rational function $f\in\set F(y)$ 
with $K$ being $\sigma_y$-standard, and 
let $r$ be a $\sigma_y$-remainder of~$S$ with respect to~$K$. 
Let $a\in \set F[y]$ be a $\sigma_y$-monic irreducible factor of $K$ over~$\set F$,
and write $K = u/v$ with $u,v\in \set F[y]$ and $\gcd(u,v) = 1$.
\begin{itemize}
\item[(i)] If $a\mid v$, then the pair
\[
(\tilde K,\tilde S) = \Big(\frac{u}{\tilde v\sigma_y(a)},aS\Big)
\quad \text{with}\ \tilde v = \frac{v}a\in\set F[y]
\]
is an RNF of~$f$ and $\tilde K$ is $\sigma_y$-standard. Moreover, there exists
a $\sigma_y$-remainder of~$\tilde S$ with respect to $\tilde K$ which has the same 
significant denominator as~$r$.

\item[(ii)] If $a\mid u$, then the pair
\[
(\tilde K,\tilde S) = \Big(\frac{\tilde u\sigma_y^{-1}(a)}{v},\sigma_y^{-1}(a)S\Big)
\quad\text{with}\ \tilde u = \frac{u}{a}\in\set F[y]
\]
is an RNF of~$f$ and $\tilde K$ is $\sigma_y$-standard. Moreover, there exists
a $\sigma_y$-remainder of~$\tilde S$ with respect to~$\tilde K$ which has the same 
significant denominator as~$r$.
\end{itemize}
\end{lemma}
\begin{proof}
(i) Assume that $a\mid v$ and let $\tilde v = v/a\in\set F[y]$.
Set $(\tilde K,\tilde S) = (u/(\tilde v\sigma_y(a)),aS)$.
Since $K$ is $\sigma_y$-standard and $a$ is a $\sigma_y$-monic factor of~$v$, we see
that $\tilde K$ is also $\sigma_y$-standard. 
The first assertion thus follows by observing that
\[
f = K\frac{\sigma_y(S)}S = \frac{u}{\tilde v a}\frac{\sigma_y(S)}S 
= \frac{u}{\tilde v\sigma_y(a)}\frac{\sigma_y(aS)}{aS}
= \tilde K\frac{\sigma_y(\tilde S)}{\tilde S}.
\]
Since $r$ is a $\sigma_y$-remainder of~$S$ with respect to~$K$, we have $S-r \in \im(\Delta_K)$,
and thus $S-r = \Delta_K(g)$ for some $g\in \set F(y)$. It follows that
$
\tilde S - ar 
=  a\Delta_K(g) 
= \Delta_{\tilde K}(a g)\in \im(\Delta_{\tilde K}).
$
It thus amounts to showing that $ar$ has a $\sigma_y$-remainder with respect to~$\tilde K$ whose
significant denominator is equal to that of~$r$. 
Write $r$ in the form $r = h+p/v$, where $h\in \set F(y)$ is proper whose denominator~$d$ is 
$\sigma_y$-normal and strongly coprime with~$K$ and $p\in \im(\phi_K)^\top$. Then 
$
ar 
= a h + p\,\sigma_y(a)/(\tilde v\sigma_y(a)).
$
Since $d$ is strongly coprime with $K$ and $a\mid v$,
it follows that $\gcd(a,d) = 1$ and thus $d$ is again the denominator of~$ah$.
By division with remainder and a subsequent application of Fact~\ref{FAC:polyred}
on the quotient part, 
we can find $\tilde g\in \set F[y]$, $\tilde h\in \set F(y)$ and $\tilde p\in \im(\phi_{\tilde K})^\top$ 
such that
\[
a r = \Delta_{\tilde K}(\tilde g) + \tilde h+ \frac{\tilde p}{\tilde v\sigma_y(a)},
\]
and $\tilde h$ is proper with denominator~$d$. 
Therefore, $\tilde h+\tilde p/(\tilde v\sigma_y(a))$ is a $\sigma_y$-remainder 
of~$ar$ (and thus $\tilde S$) with respect to $\tilde K$ which has the same 
significant denominator as~$r$.

\medskip\noindent
(ii) The assertions can be shown by following similar lines as part (i).
\end{proof}

Now we are ready to extend Proposition~\ref{PROP:onekernel} to the case of different kernels.
\begin{proposition}\label{PROP:twokernels}
Let $(K,S)$ and $(\tilde K,\tilde S)$ be two RNFs of a nonzero rational function $f\in\set F(y)$
with $K,\tilde K$ being $\sigma_y$-standard.
Let $r$ be a $\sigma_y$-remainder of~$S$ with respect to~$K$ and 
$\tilde r$ be a $\sigma_y$-remainder of $\tilde S$ with respect to~$\tilde K$. 
Then the significant denominators of~$r$ and $\tilde r$ are $\sigma_y$-related.
\end{proposition}
\begin{proof}
Let $d$ and $\tilde d$ be the significant denominators of~$r$ and~$\tilde r$, respectively.
Let
\[
K = K_sK_n\quad\text{and}\quad \tilde K = \tilde K_s\tilde K_n
\]
be $\sigma_y$-splitting factorizations of~$K$ and $\tilde K$, respectively,
with $K_n,\tilde K_n$ being $\sigma_y$-monic.
Since $(K,S)$ and $(\tilde K,\tilde S)$ are both RNFs of~$f$, by Lemma~\ref{LEM:twornfs}~(i),
\[
\frac{\tilde K_s}{K_s} = 
\begin{cases}
1 & \text{in the usual shift case},\\[1ex]
q^m \ \text{for some}\ m\in \set Z & \text{in the $q$-shift case}.
\end{cases}
\]
Note that in the $q$-shift case, we may assume without loss of generality that $m\geq 0$, for, otherwise
the roles of~$K$ and $\tilde K$ can be switched. With the help of Lemma~\ref{LEM:specialfac},
we find that, in either case, there exists an RNF $(\bar K,\bar S)$ of~$f$ such that
the kernel $\bar K$ is $\sigma_y$-standard and admits a $\sigma_y$-splitting factorization 
\[
\bar K = \bar K_s \bar K_n, \quad\text{where $\bar K_s = K_s$ and $\bar K_n = \tilde K_n$},
\]
and the shell $\bar S$ has a $\sigma_y$-remainder with respect to~$\bar K$ of significant 
denominator~$\tilde d$.

By Lemma~\ref{LEM:twornfs}~(ii), there exists a one-to-one correspondence $\varphi$ between 
the multi-sets of $\sigma_y$-monic irreducible factors of the denominators of~$K_n$ and $\bar K_n$ over~$\set F$
such that $a$ and $\varphi(a)$ are $\sigma_y$-equivalent for all $a\in\set F[y]$ dividing the denominator of~$K_n$.
If $K_n\in \set F[y]$, then so is~$\bar K_n$, implying that $K$ and $\bar K$ have 
the same denominator as $K_n$ and $\bar K_n$ are both $\sigma_y$-monic. 
Now assume that $K_n\notin \set F[y]$ and 
let $a\in\set F[y]$ be a $\sigma_y$-monic irreducible factor of the denominator of~$K_n$ over~$\set F$.
Then $\varphi(a) = \sigma_y^\ell(a)$ for some $\ell\in \set Z$. 
If $\ell <0$, then a repeated use of Lemma~\ref{LEM:normalfac}~(i) on $(\bar K,\bar S)$
leads to a new RNF $(K',S')$ of $f$ such that $K' = \bar K \sigma_y^\ell(a)/a$ is 
$\sigma_y$-standard and $S'$ has a $\sigma_y$-remainder with respect to~$K'$
of significant denominator~$\tilde d$. 
Analogously, in the case where~$\ell>0$, we can apply Lemma~\ref{LEM:normalfac} (i) to $(K,S)$
repeatedly
and obtain a new RNF $(K',S')$ satisfying that $K'=Ka/\sigma_y^\ell(a)$ is $\sigma_y$-standard
and $S'$ has a 
$\sigma_y$-remainder with respect to~$K'$ of significant denominator~$d$.
Note that, compared with the original kernels $K$ and~$\bar K$, the new kernel $K'$ and 
the remaining one have one more $\sigma_y$-monic irreducible factor over~$\set F$
(counting the multiplicities) in common.

Applying the above argument to each $\sigma_y$-monic irreducible factor of the denominator of~$K_n$ over~$\set F$
and using Lemma~\ref{LEM:normalfac} (ii) for the numerator in the same fashion, 
we finally obtain two
new RNFs of $f$ where kernels are equal and $\sigma_y$-standard, and shells have respective 
$\sigma_y$-remainders of significant denominators $d$ and~$\tilde d$. 
It follows from Proposition~\ref{PROP:onekernel} that $d$ and $\tilde d$ are $\sigma_y$-related.
\end{proof}

It can be shown that the application of the operator~$\sigma_x$ preserves 
RNFs with $\sigma_y$-standard kernels, as well as associated 
$\sigma_y$-remainders of the shells.
\begin{lemma}\label{LEM:xshiftrnf}
Let $(K,S)$ be an RNF of a nonzero rational function $f\in\set F(y)$ 
with $K$ being $\sigma_y$-standard, and
let $r$ be a $\sigma_y$-remainder of~$S$ with respect to~$K$ with significant denominator~$d$. 
Then for every $i\in \set N$,
the pair $(\sigma_x^i(K),\sigma_x^i(S))$ is an RNF of~$\sigma_x^i(f)$, 
$\sigma_x^i(K)$ is $\sigma_y$-standard, and
$\sigma_x^i(r)$ is a $\sigma_y$-remainder of $\sigma_x^i(S)$ with respect to~$\sigma_x^i(K)$
with significant denominator~$\sigma_x^i(d)$.
\end{lemma}
\begin{proof}
Let $i\in \set N$. 
Since $(K,S)$ is an RNF of~$f$, we have $f = K\sigma_y(S)/S$ and thus
\[
\sigma_x^i(f) = \sigma_x^i(K)\sigma_x^i\Big(\frac{\sigma_y(S)}S\Big) 
= \sigma_x^i(K) \frac{\sigma_y(\sigma_x^i(S))}{\sigma_x^i(S)}.
\]
Notice that for any $a,b\in \set F[y]$, we have $\gcd(a,b) = 1$ if and only if 
$\gcd(\sigma_x^i(a),\sigma_x^i(b)) =1$. Also notice that in the $q$-shift case,
$\sigma_x^i(a)\mid_{y=0} = \sigma_x^i(a\mid_{y=0})$ for any $a\in \set F[y]$. 
We thus obtain from the $\sigma_y$-standardness
of~$K$ that $\sigma_x^i(K)$ is $\sigma_y$-standard, and then 
$(\sigma_x^i(K),\sigma_x^i(S))$ is an RNF of~$\sigma_x^i(f)$.

Since $r$ is a $\sigma_y$-remainder of~$S$ with respect to~$K$, we have $S-r \in \im(\Delta_K)$,
and thus $S-r =\Delta_K(g)$ for some $g\in \set F(y)$. It follows that
$
\sigma_x^i(S) - \sigma_x^i(r) 
=\sigma_x^i(\Delta_K(g)) 
= \Delta_{\sigma_x^i(K)}(\sigma_x^i(g))\in \im(\Delta_{\sigma_x^i(K)}).
$
It remains to show that $\sigma_x^i(r)$ is a $\sigma_y$-remainder with respect to~$\sigma_x^i(K)$
and has the significant denominator~$\sigma_x^i(d)$.
Write $r = h+p/v$, where $h\in \set F(y)$ is proper whose denominator~$d$
is $\sigma_y$-normal and strongly coprime with~$K$, $p\in \im(\phi_K)^\top$, and $v$ 
is the denominator of~$K$. Then $\sigma_x^i(r) = \sigma_x^i(h) + \sigma_x^i(p)/\sigma_x^i(v)$.
Clearly, $\sigma_x^i(h)\in \set F(y)$ is again proper whose denominator 
$\sigma_x^i(d)$ is $\sigma_y$-normal and strongly coprime with~$\sigma_x^i(K)$. 
In order to show that $\sigma_x^i(p)\in \im(\phi_{\sigma_x^i(K)})^\top$, we observe that 
$a\in \im(\phi_K)$ if and only if $\sigma_x^i(a)\in\im(\phi_{\sigma_x^i(K)})$ for any 
$a\in \set F[y]$. 
Since $\sigma_x^i\circ\deg_y=\deg_y\circ\, \sigma_x^i$, the images of all
elements from an echelon basis (namely an $\set F$-basis in which different
elements have distinct $y$-degrees) for $\im(\phi_K)$ under the automorphism~$\sigma_x^i$
form an echelon basis for $\im(\phi_{\sigma_x^i(K)})$. This implies that $\im(\phi_K)^\top$ 
and $\im(\phi_{\sigma_x^i(K)})^\top$ share the same echelon basis.
It thus follows from $p\in \im(\phi_K)^\top$ that $\sigma_x^i(p)\in \im(\phi_{\sigma_x^i(K)})^\top$. 
Accordingly, $\sigma_x^i(r)$ is a $\sigma_y$-remainder of~$\sigma_x^i(S)$ with
respect to~$\sigma_x^i(K)$ with significant denominator~$\sigma_x^i(d)$.
\end{proof}
The desired relation is given below.
\begin{proposition}\label{PROP:shiftrem}
Let $T$ be a $(\sigma_x,\sigma_y)$-hypergeometric term whose $\sigma_y$-quotient
has a $\sigma_y$-standard kernel $K$ and a corresponding shell~$S$, and set $H = T/S$. 
Then for every $i\in \set N$, there exists a rational function $g_i\in \set F(y)$
and a $\sigma_y$-remainder $r_i$ with respect to~$K$ such that
\begin{equation}\label{EQ:ithhyperred}
\sigma_x^i(T) = \Delta_y(g_iH) + r_i H.
\end{equation}
Furthermore, let $r$ be a $\sigma_y$-remainder of~$S$ with respect to~$K$ and 
$d$ be its significant denominator. Then for every $i\in \set N$, the significant 
denominator of~$r_i$ in \eqref{EQ:ithhyperred} is $\sigma_y$-related to~$\sigma_x^i(d)$.
\end{proposition}
\begin{proof}
Let $i\in \set N$, $f = \sigma_y(T)/T$ and $N = \sigma_x(H)/H$. Then $f,N\in \set F(y)$.
Since $T = SH$, we have $K = \sigma_y(H)/H$ and
\[
\sigma_x^i(T) = \sigma_x^i(SH) 
= \sigma_x^i(S) \sigma_x^{i-1}(N)\cdots \sigma_x(N)N H
=\tilde S H,
\]
where $\tilde S = \sigma_x^i(S) \sigma_x^{i-1}(N)\cdots \sigma_x(N)N$.
It follows that
\[
\sigma_x^i(f) = \sigma_x^i\Big(\frac{\sigma_y(T)}T\Big)
= \frac{\sigma_y(\sigma_x^i(T))}{\sigma_x^i(T)}
= \frac{\sigma_y(\tilde SH)}{\tilde SH}
= K\frac{\sigma_y(\tilde S)}{\tilde S}.
\]
Notice that $K$ is $\sigma_y$-standard. So $(K,\tilde S)$ is an RNF of~$\sigma_x^i(f)$.
This implies that $\sigma_x^i(T)$ can be viewed as a 
$(\sigma_x,\sigma_y)$-hypergeometric term whose $\sigma_y$-quotient 
has a $\sigma_y$-standard kernel $K$ and a corresponding shell~$\tilde S$.
Applying Theorem~\ref{THM:shellred} to $\sigma_x^i(T)$ gives 
\eqref{EQ:ithhyperred}, where $g_i\in \set F(y)$ and $r_i$ is a 
$\sigma_y$-remainder with respect to~$K$.
Thus the first assertion follows.

For proving the second assertion, we note that $(K,S)$ is an RNF of~$f$ with
$K$ being $\sigma_y$-standard. Since $r$ is a $\sigma_y$-remainder of~$S$
with respect to~$K$ whose significant denominator is~$d$, by 
Lemma~\ref{LEM:xshiftrnf}, the pair $(\sigma_x^i(K),\sigma_x^i(S))$ is 
an RNF of~$\sigma_x^i(f)$, $\sigma_x^i(K)$ is $\sigma_y$-standard, and 
$\sigma_x^i(r)$ is a $\sigma_y$-remainder of $\sigma_x^i(S)$ with respect 
to $\sigma_x^i(K)$ whose significant denominator is~$\sigma_x^i(d)$.
Observe that $\sigma_x^i(T) = \tilde S H$ and $K = \sigma_y(H)/H$. So
$\tilde S = \Delta_K(g_i) + r_i$ by \eqref{EQ:ithhyperred},
implying that $r_i$ is a $\sigma_y$-remainder of~$\tilde S$ with respect to~$K$. 
Since $(K,\tilde S)$ is also an RNF of $\sigma_x^i(f)$, we conclude from
Proposition~\ref{PROP:twokernels} that the significant denominator of~$r_i$ is 
$\sigma_y$-related to~$\sigma_x^i(d)$.
\end{proof}

\section{Upper and lower order bounds}\label{SEC:ord}
With the notation introduced in Proposition~\ref{PROP:shiftrem}, by knowing the connection 
between each~$r_i$ in~\eqref{EQ:ithhyperred} and the initial $\sigma_y$-remainder~$r$,
we proceed to arguing in this section that these $r_i$ can be chosen in such a way 
that all their significant denominators divide a pre-specified a priori multiple
provided that the significant denominator of~$r$ is integer-linear, 
leading to upper and lower bounds on the order of telescopers.

We first recall the notion of integer-linear rational functions. 
\begin{definition}\label{DEF:intlinear}
An irreducible polynomial $p$ in~$\set K[x,y]$ is said to be {\em integer-linear} 
(over~$\set K$) if there exist $m,n\in \set Z$, not both zero, such that $p$ and 
$\sigma_x^m\sigma_y^n(p)$ are associates.
A polynomial in~$\set K[x,y]$ is said to be {\em integer-linear} (over~$\set K$) 
if all its irreducible factors are integer-linear. 
A rational function in~$\set F(y)$ is said to be {\em integer-linear} 
(over~$\set K$) if its denominator and numerator are both integer-linear.
\end{definition}
We refer to \cite{GHLZ2019, GHLZ2021} for algorithms determining the 
integer-linearity of rational functions.  The ``integer-linear" attribute in the 
name is justified by the following proposition. 
\begin{proposition}[{\cite[Proposition~5.4]{CDGHL2025}}]\label{PROP:intlinear}
Let $p$ be an irreducible polynomial in~$\set K[x,y]$.
\begin{itemize}
\item[(i)] In the usual shift case, $p$ is integer-linear if and only if
$p(x,y) = P(\lambda x+ \mu y)$ for some $P(z)\in \set K[z]$ and 
$\lambda,\mu\in \set Z$ not both zero.

\item[(ii)] In the $q$-shift case, $p$ is integer-linear if and only if
$p(x,y) = x^\alpha y^\beta P(x^\lambda y^\mu)$ for some 
$P(z)\in \set K[z]$ and $\alpha,\beta,\lambda,\mu\in \set Z$ 
with $\lambda, \mu$ not both zero.
\end{itemize}
\end{proposition}
Let $p$ be an integer-linear irreducible polynomial in~$\set K[x,y]$.
By Proposition~\ref{PROP:intlinear}, there exists a pair 
$(P(z),\{\lambda,\mu\})\in \set K[z]\times\set Z^2$ satisfying 
\begin{equation}\label{EQ:ilform}
p(x,y) = P(\lambda x + \mu y) \quad \big(\text{resp.}\ p(x,y) = x^\alpha y^\beta P(x^\lambda y^\mu)\big)
\end{equation}
in the usual shift (resp.\ $q$-shift) case, where $\alpha,\beta\in \set Z$ can be uniquely determined 
once the pair $(P(z),\{\lambda,\mu\})$ is fixed.
Such a pair becomes unique if we further impose the following conditions:
\begin{itemize}
\item[(i)] $\gcd(\lambda,\mu) = 1$ and $\mu\geq 0$;
\item[(ii)] $(\lambda,\mu) = (1,0)$ if $p\in \set K[x]$ and $(\lambda,\mu) = (0,1)$ if $p\in \set K[y]$;
\item[(iii)] $P(0) \neq 0$ in the $q$-shift case,
\end{itemize}
and it will be referred to as the {\em univariate representation} of~$p$. 
Now assume that $(P(z),\{\lambda,\mu\})$ is the univariate representation of~$p$.
We remark that in the $q$-shift case, $x^\alpha y^\beta$ appearing in \eqref{EQ:ilform} will 
be the trailing monomial of~$p$ with respect to the pure lexicographic order $x\prec y$,
and thus we have $\alpha,\beta\geq 0$.
Since $\gcd(\lambda,\mu) = 1$ and $\mu\geq 0$, there exist unique integers $s,t$
satisfying the constraints that $0\leq s<\mu$ and $|t|<|\lambda|$ if $\lambda\mu\neq 0$,
or $(s,t) = (1,0)$ if $\mu = 0$, or $(s,t) = (0,1)$ if $\lambda = 0$, such that $s\lambda + t\mu = 1$. 
Define the operator 
\[
\delta_{\lambda,\mu}=\sigma_x^s\sigma_y^t.
\]
Then $\delta_{\lambda,\mu}(P(z)) = \sigma_z(P(z))$, where $z = \lambda x + \mu y$ in the usual shift case 
and $z = x^\lambda y^\mu$ in the $q$-shift case, and thus it allows us to treat integer-linear polynomials 
as univariate ones. Moreover, we make frequent use in the future of the trivial relations
\[
\delta_{\lambda,\mu}^\lambda(p) =c_1 \sigma_x(p)
\quad\text{and}\quad
\delta_{\lambda,\mu}^\mu(p) = c_2 \sigma_y(p)
\quad\text{for some}\ c_1,c_2\in \set K.
\]
For a Laurent polynomial $\alpha = \sum_{i=m}^n k_i\delta_{\lambda,\mu}^i$
in $\set Z[\delta_{\lambda,\mu},\delta_{\lambda,\mu}^{-1}]$, we similarly define
\[
p^\alpha = \prod_{i=m}^n \delta_{\lambda,\mu}^i(p)^{k_i}.
\]

We now extend the notion of $\sigma_y$-equivalence to the bivariate case. 
Let $a,b$ be two irreducible polynomials in~$\set K[x,y]$. We say that $a$ and $b$ are 
{\em $(\sigma_x,\sigma_y)$-equivalent} if $a$ is an associate of $\sigma_x^m\sigma_y^n(b)$ for some 
integers $m,n$.
This is again an equivalence relation. Further assume that $a$ and $b$ are both integer-linear
whose univariate representations are given by $(P(z),\{\lambda,\mu\})$ and $(Q(z),\{\nu,\omega\})$, 
respectively. It is then not hard to see that $a,b$ are $(\sigma_x,\sigma_y)$-equivalent if and only if
$(\lambda,\mu) = (\nu,\omega)$ and $P(z),Q(z)$ are $\sigma_z$-equivalent, and in this case
$a$ is an associate of $\delta_{\lambda,\mu}^\ell(b)$ for some integer~$\ell$.

Let $d$ be a $\sigma_y$-normal and integer-linear polynomial in~$\set F[y]$. 
By grouping together the factors in the same $(\sigma_x,\sigma_y)$-equivalence class, 
we can decompose $d$ as
\begin{equation}\label{EQ:ildecomp}
d = c\,p_1^{\alpha_1}\cdots p_m^{\alpha_m},
\end{equation}
where $c\in \set F$, $m\in \set N$, each $p_i\in\set K[x,y]$ is $\sigma_y$-normal, integer-linear, 
irreducible, and has univariate representation $(P_i(z),\{\lambda_i,\mu_i\})$ with $\mu_i>0$ and $P_i(z)$ being 
$\sigma_z$-monic, and each $\alpha_i\in \set N[\delta_{\lambda_i,\mu_i},\delta_{\lambda_i,\mu_i}^{-1}]\setminus\{0\}$. 
Moreover, the $p_i$ are pairwise $(\sigma_x,\sigma_y)$-inequivalent. 
Similar to the univariate case, the factorization \eqref{EQ:ildecomp} is not unique; 
however, the components $p_i^{\alpha_i}$ are uniquely determined by~$d$ 
since $\set F[y]$ is a unique factorization domain.

Recall that for two integers $\lambda,\mu$ with $\mu\neq 0$, the {\em remainder of~$\lambda$ modulo~$\mu$},
and the {\em quotient of~$\lambda$ by~$\mu$} are respectively defined as the unique integers $i$ and $j$
such that $\lambda = \mu j+i$ and $0\leq i<|\mu|$.
For convenience, we sometimes denote the remainder~$i$ by $(\lambda \mod \mu)$.

In order to describe a multiple of significant denominators of $\sigma_y$-remainders, 
we introduce below several lemmas.
The first lemma is an extended version of~\cite[Lemma~5.3]{Huan2016}.
\begin{lemma}\label{LEM:stronglycoprime}
Let $K\in \set F(y)$ be $\sigma_y$-reduced and let $p\in \set F(y)$ be 
$\sigma_y$-normal and irreducible over~$\set F$. Then there exists an integer 
$m$ such that $\sigma_y^m(p)$ is strongly coprime with~$K$.
\end{lemma}
\begin{proof}
Let $u$ and $v$ be the numerator and denominator of~$K$, respectively.
The following three cases can be distinguished.

\smallskip\noindent
{\em Case~1.} $\sigma_y^\ell(p)\nmid u$ and $\sigma_y^\ell(p)\nmid v$ for all
integers~$\ell$. Then $p$ is strongly coprime with~$K$ by definition. 
Letting $m = 0$ yields the lemma.

\smallskip\noindent
{\em Case~2.} $\sigma_y^k(p)\mid u$ for some integer~$k$. Since $K$ is 
$\sigma_y$-reduced, $\gcd(\sigma_y^\ell(p),v) =1$ for every integer~$\ell$.
Set $m = \max\{i\in \set N\mid \sigma_y^i(p)\mid u\}+1$, which is a finite 
integer as $p$ is $\sigma_y$-normal. By definition, $\sigma_y^m(p)$ is 
strongly coprime with~$K$.

\smallskip\noindent
{\em Case~3.} $\sigma_y^k(p)\mid v$ for some integer~$k$. Since $K$ is 
$\sigma_y$-reduced, $\gcd(\sigma_y^\ell(p),u) =1$ for every integer~$\ell$.
Set $m = \min\{i\in \set N\mid \sigma_y^i(p)\mid v\}-1$, which is a finite 
integer as $p$ is $\sigma_y$-normal. By definition, $\sigma_y^m(p)$ is 
strongly coprime with~$K$.
\end{proof}

The next lemma indicates that there are finitely many $\sigma_y$-equivalence classes 
produced by shifting a $\sigma_y$-normal, integer-linear and irreducible polynomial in~$\set K[x,y]$
as a univariate one.
\begin{lemma}\label{LEM:shiftfinite}
Let $p\in\set K[x,y]$ be $\sigma_y$-normal, integer-linear, irreducible and of univariate representation 
$(P(z),\{\lambda,\mu\})$ with $\mu >0$. Then the $\mu$ polynomials
\[
p,\, p^{\delta_{\lambda,\mu}},\, \dots, \, p^{\delta_{\lambda,\mu}^{\mu-1}}
\]
are mutually $\sigma_y$-inequivalent. 
Moreover, $p^{\delta_{\lambda,\mu}^\ell}$ is $\sigma_y$-equivalent to $p^{\delta_{\lambda,\mu}^{(\ell\mod \mu)}}$
for all $\ell\in\set Z$.
\end{lemma}
\begin{proof}
There is nothing to show if $\mu = 1$. Now assume that $\mu>1$. Then $p\notin \set K[x]\cup\set K[y]$.
In order to verify the first assertion, it amounts to showing that $p$ is $\sigma_y$-inequivalent to 
all the $p^{\delta_{\lambda,\mu}^i}$ for $i=1,\dots,\mu-1$.
Suppose that there exists an integer $i$ with $1\leq i\leq \mu-1$ such that $p$ and 
$p^{\delta_{\lambda,\mu}^i}$ are $\sigma_y$-equivalent. Then $p$ is an associate of 
$\sigma_y^\ell(p^{\delta_{\lambda,\mu}^i})$
for some integer~$\ell$. Since $(P(z),\{\lambda,\mu\})$ is the univariate representation of~$p$, 
it follows that $P(z)\mid \sigma_z^{i+\mu \ell}(P(z))$. Observe that $i+\mu\ell\neq 0$ as $1\leq i\leq \mu-1$. 
Thus $P(z)\in\set K[z]$ is $\sigma_z$-special. 
By Lemma~\ref{LEM:special}, either $P(z)\in \set K$ in the usual shift case or 
$P(z)/z^k\in \set K$ for some integer~$k$ in the $q$-shift case. 
Notice that $P(0) \neq 0$ in the $q$-shift case. So we have $P(z)\in \set K$ in either case.
Thus either $p\in \set K$ in the usual shift case or $p = c\,x^\alpha y^\beta$ for $c\in \set F$ and 
$\alpha,\beta\in \set N$ in the $q$-shift case, leading to a contradiction that $p\in \set K[x]\cup\set K[y]$
as $p$ is irreducible.
Therefore, the first assertion holds.

Let $\ell\in \set Z$ and let $i,j$ be the remainder of~$\ell$ modulo~$\mu$ and the quotient of~$\ell$ by~$\mu$, respectively.
Then $\ell = \mu j + i$ and $0\leq i\leq \mu-1$. Since $(P(z),\{\lambda,\mu\})$ is the univariate representation of~$p$, 
we see that $p^{\delta_{\lambda,\mu}^{\ell}}$ is an associate of $\sigma_y^j(p^{\delta_{\lambda,\mu}^i})$, which means that
$p^{\delta_{\lambda,\mu}^{\ell}}$ is $\sigma_y$-equivalent to~$p^{\delta_{\lambda,\mu}^i}$.
\end{proof}
We derive below a ``local" multiple, which constitutes an important ingredient for the general case
described in the succeeding proposition.
\begin{lemma}\label{LEM:localcommonden}
Let $K\in\set F(y)$ be $\sigma_y$-reduced, let $p\in\set K[x,y]$ be integer-linear, irreducible and 
of univariate representation $(P(z),\{\lambda,\mu\})$ with $\mu>0$, and let 
$\alpha \in \set N[\delta_{\lambda,\mu},\delta_{\lambda,\mu}^{-1}]\setminus\{0\}$ 
with maximum coefficient $k$ with respect to~$\delta_{\lambda,\mu}$.
Assume that $p^\alpha$ is $\sigma_y$-normal and strongly coprime with~$K$.
Then there exists a multiple $D\in \set K[x,y]$ of~$p^\alpha$ which is $\sigma_y$-normal, strongly coprime with~$K$,
and $\sigma_y$-related to the polynomial 
\[
p^{k(1-\delta_{\lambda,\mu}^\mu)/(1-\delta_{\lambda,\mu})}.
\]
\end{lemma}
\begin{proof}
Let $\Lambda$ be the support of~$\alpha$, that is, the set of indices $i\in\set Z$ with the property that 
the coefficient of $\delta_{\lambda,\mu}^i$ 
in $\alpha$ is nonzero, and set
\[
\bar \Lambda = \{0,1,\dots,\mu-1\}\setminus\{(i \mod \mu)\mid i\in \Lambda\}.
\]
Because $p^\alpha$ is $\sigma_y$-normal, so is $p^{\delta_{\lambda,\mu}^i}$ for any $i\in \set Z$.
For each $j\in \bar \Lambda$, by Lemma~\ref{LEM:stronglycoprime}, there exists an integer $\ell_j$ such that 
$\sigma_y^{\ell_j}(p^{\delta_{\lambda,\mu}^j})$ (and thus $p^{\delta_{\lambda,\mu}^{j+\mu\ell_j}}$) is
strongly coprime with~$K$.
Since $\alpha \in \set N[\delta_{\lambda,\mu},\delta_{\lambda,\mu}^{-1}]\setminus\{0\}$,
we have $k>0$. Define
\[
\beta = k \Big(\sum_{i\in \Lambda}\delta_{\lambda,\mu}^i 
+ \sum_{j\in\bar\Lambda}\delta_{\lambda,\mu}^{j+\mu\ell_j}\Big)
\]
and $D = p^\beta$. Clearly, $\beta\in \set N[\delta_{\lambda,\mu},\delta_{\lambda,\mu}^{-1}]\setminus\{0\}$,
$D\in \set K[x,y]$ and $p^\alpha \mid D$. 
Since $p^\alpha$ and the $p^{\delta_{\lambda,\mu}^{j+\mu\ell_j}}$ are strongly coprime with~$K$, so is~$D$. 
Due to the $\sigma_y$-normality of $p^\alpha$, we see from 
Lemma~\ref{LEM:shiftfinite} that $\Lambda\cup\bar\Lambda$
has exactly $\mu$ elements and different elements have distinct remainders modulo~$\mu$.
Again, by Lemma~\ref{LEM:shiftfinite}, there is a one-to-one correspondence $\varphi$ between the two sets
\[
\Big\{p^{\delta_{\lambda,\mu}^i}\Big\}_{i\in \Lambda}
\cup\ \Big\{p^{\delta_{\lambda,\mu}^{j+\mu\ell_j}}\Big\}_{j\in\bar\Lambda}
\quad\text{and}\quad
\Big\{p,\, p^{\delta_{\lambda,\mu}},\, \dots, \, p^{\delta_{\lambda,\mu}^{\mu-1}}\Big\},
\]
such that $a$ and $\varphi(a)$ are $\sigma_y$-equivalent for all $a$ in the former set.
By the $\sigma_y$-inequivalence of the~$\mu$ polynomials 
$p,\, p^{\delta_{\lambda,\mu}},\, \dots, \, p^{\delta_{\lambda,\mu}^{\mu-1}}$,
we obtain that $D$ and $p^{k(1-\delta_{\lambda,\mu}^\mu)/(1-\delta_{\lambda,\mu})}$ are both $\sigma_y$-normal,
and thus they are $\sigma_y$-related.
\end{proof}
The following 
lemma can be extracted from the proof of~\cite[Theorem~4.1]{CDGHL2025}.
\begin{lemma}\label{LEM:adjustrem}
Let $K\in \set F(y)$ be $\sigma_y$-standard, let $d\in \set F[y]$ be $\sigma_y$-normal
and strongly coprime with~$K$, and let $s$ be a $\sigma_y$-remainder with respect 
to~$K$. Then there exists a $\sigma_y$-remainder $t$ with respect to $K$ such that
$s-t\in \im(\Delta_K)$ and the significant denominator of~$t$ is $\sigma_y$-coprime
with~$d$.
\end{lemma}
We are now ready to obtain the a priori multiple as promised at the beginning of this section.
\begin{proposition}\label{PROP:commonden}
Let $T$ be a $(\sigma_x,\sigma_y)$-hypergeometric term whose $\sigma_y$-quotient
has a $\sigma_y$-standard kernel $K$ and a corresponding shell~$S$,
and let $r$ be a $\sigma_y$-remainder of~$S$ with respect to~$K$ whose
significant denominator $d$ is integer-linear and admits the factorization \eqref{EQ:ildecomp}.
For all $j=1,\dots,m$, set $k_j$ to be the maximum coefficient of~$\alpha_j$ with respect 
to~$\delta_{\lambda,\mu}$, and define
\[
D_0 = \prod_{j=1}^mp_j^{k_j(1-\delta_{\lambda_j,\mu_j}^{\mu_j})/(1-\delta_{\lambda_j,\mu_j})}.
\]
Then $D_0\in\set K[x,y]$ is $\sigma_y$-normal, and there exists a multiple $D\in \set F[y]$ of~$d$ which is 
$\sigma_y$-normal, strongly coprime with~$K$, and $\sigma_y$-related to~$D_0$
such that \eqref{EQ:ithhyperred} holds for every $i\in \set N$, where $H = T/S$, $g_i\in \set F(y)$ and 
$r_i$ is a $\sigma_y$-remainder with respect to~$K$ whose significant denominator~$d_i$ divides~$D$.
In particular, $d_0$ and $d$ are associates.
\end{proposition}
\begin{proof}
The $\sigma_y$-normality of~$D_0$ is evident by Lemma~\ref{LEM:shiftfinite}
and the $(\sigma_x,\sigma_y)$-inequivalence of the~$p_j$.
Because $d$ is $\sigma_y$-normal and strongly coprime with~$K$, 
so are the $p_j^{\alpha_j}$ by Proposition~\ref{PROP:props}~(i).
For all $j = 1,\dots,m$, applying Lemma~\ref{LEM:localcommonden} to $p_j^{\alpha_j}$ in~\eqref{EQ:ildecomp}
delivers a multiple $D_j\in\set K[x,y]$ of~$p_j^{\alpha_j}$ which is $\sigma_y$-normal, strongly coprime with~$K$,
and $\sigma_y$-related to the polynomial 
\[
p_j^{k_j(1-\delta_{\lambda_j,\mu_j}^{\mu_j})/(1-\delta_{\lambda_j,\mu_j})}.
\]
Let $D = D_1\cdots D_m$. Then $D\in\set K[x,y]$. 
Since the $p_j$ are pairwise $(\sigma_x,\sigma_y)$-inequivalent, 
the~$D_j$ are mutually coprime and also $\sigma_y$-coprime.
Thus $d\mid D$ in $\set F[y]$ and $D$ is again $\sigma_y$-normal, strongly coprime with~$K$, 
and $\sigma_y$-related to~$D_0$. 

Let $i\in \set N$. By Proposition~\ref{PROP:shiftrem}, there exists a rational function $g_i\in \set F(y)$
and a $\sigma_y$-remainder $r_i$ with respect to~$K$ such that \eqref{EQ:ithhyperred} holds.
Without loss of generality, we may further assume that the significant denominator $d_i$ of~$r_i$ is 
$\sigma_y$-coprime with~$D$, for, otherwise, with Lemma~\ref{LEM:adjustrem} applied to $D$ and~$r_i$, 
we can always replace $r_i$ in \eqref{EQ:ithhyperred} by one with the required property. 
We claim that such a $d_i$ actually divides $D$ in~$\set F[y]$. 
The claim is trivial if $d_i\in \set F$. 

Now assume that $d_i\notin \set F$ and let $a_i\in\set F[y]$ be an irreducible factor of~$d_i$ over~$\set F$. 
Since $d_i$ is $\sigma_y$-related to $\sigma_x^i(d)$ by Proposition~\ref{PROP:shiftrem},
there exists an irreducible factor $a\in\set F[y]$ of~$d$ over~$\set F$ such that 
$a_i$ is $\sigma_y$-equivalent to~$\sigma_x^i(a)$.
It follows from the factorization \eqref{EQ:ildecomp} that 
\[
a\ \text{is an associate of}\ p_j^{\delta_{\lambda_j,\mu_j}^\ell}
\quad\text{and thus}\quad
\sigma_x^i(a)\ \text{is an associate of}\ 
p_j^{\delta_{\lambda_j,\mu_j}^{\ell+\lambda_j i}},
\]
where $j,\ell\in\set Z$ with $1\leq j\leq m$. 
This implies that $a_i$ is $\sigma_y$-equivalent to $p_j^{\delta_{\lambda_j,\mu_j}^{\ell+\lambda_j i}}$,
which in turn, by Lemma~\ref{LEM:shiftfinite}, is $\sigma_y$-equivalent to an irreducible factor of~$D_0$ over~$\set F$.
Since $D_0$ is $\sigma_y$-related to~$D$, there exists an irreducible factor of~$D$ over~$\set F$ which is 
$\sigma_y$-equivalent to~$a_i$. 
Due to the $\sigma_y$-coprimeness of~$d_i$ and $D$, we conclude that $a_i\mid D$. 
By the arbitrariness of~$a_i$, we have $d_i\mid D$. 
In particular, $d_0\mid D$. Since $d_0$ is $\sigma_y$-related to $d$, $d\mid D$, 
and $D$ is $\sigma_y$-normal,
we see that $d_0$ and~$d$ only differ by a nonzero element in~$\set F$, that is, they are associates.
\end{proof}

We depict in the next theorem upper and lower order bounds on telescopers.
Note that in the usual shift case, the following assertions (ii) and (iii) coincide with 
Theorems~5.5 and 5.6 in~\cite{Huan2016}, respectively.
\begin{theorem}\label{THM:bounds}
With the notation introduced in Proposition~\ref{PROP:commonden}, 
the following assertions hold.
\begin{itemize}
\item[(i)] A nonzero operator $L = \sum_{i=0}^\rho \ell_i S_x^i\in \set F[S_x]$ is a telescoper for~$T$
if and only if $\sum_{i=0}^\rho \ell_i r_ i = 0$.
\item[(ii)] Telescopers for~$T$ exist and their minimal order is no more than 
$\dim(\im(\phi_K)^\top)+\deg_y(D_0)$,
which, by writing $K = u/v$ with $u,v\in\set F[y]$ and $\gcd(u,v) = 1$, is equal to
\[
\max\{\deg_y(u),\deg_y(v)\}-\llbracket0\leq \deg_y(u-v)\leq \deg_y(u)-1\rrbracket
+ \sum_{j=1}^mk_j\mu_j \deg_z(P_j(z))
\]
in the usual shift case, or
\[
\max\{\deg_y(u),\deg_y(v)\}
+\llbracket K\ \text{is a nonpositive power of $q$}\rrbracket
+ \sum_{j=1}^mk_j\mu_j^2\deg_z(P_j(z))
\]
in the $q$-shift case.
\item[(iii)] Assume that $r$ is nonzero.
For all $i = 1,\dots,m$, let $\Lambda_i$ be the support of~$\alpha_i$, and write 
$\alpha_i = \sum_{j\in \Lambda_i} k_{ij}\delta_{\lambda_i,\mu_i}^j$ with $k_{ij}\in \set Z^+$. 
Then the order of a telescoper for~$T$ is at least
\[
\max_{1\leq i\leq m,\,j\in \Lambda_i}\ \min
\{\rho\in\set Z^+\mid k_{ij}\leq k_{i\ell} \ \text{and}\ 
j \equiv \ell + \lambda_i\rho \mod \mu_i\ \text{for some}\ \ell\in \Lambda_i\}.
\]
\end{itemize}
\end{theorem}
\begin{proof}
(i) Let $L = \sum_{i=0}^\rho \ell_i S_x^i$ with $\rho \in \set N$ and $\ell_0,\dots,\ell_\rho\in \set F$, not all zero.
Using \eqref{EQ:ithhyperred} with $i=0,\dots,\rho$, we have
\begin{equation}\label{EQ:telescoperred}
L(T) = \Delta_y\Big(\sum_{i=0}^\rho \ell_ig_i H\Big) + \Big(\sum_{i=0}^\rho \ell_i r_i\Big) H.
\end{equation}
For $i = 0,\dots, \rho$, since $d_i$ divides~$D$, we can write $r_i$ in the form
\begin{equation}\label{EQ:ithrem}
r_i = \frac{a_i}D + \frac{p_i}v,
\end{equation}
where $a_i\in \set F[y]$ with $\deg_y(a_i) < \deg_y(D)$, $p_i\in \im(\phi_K)^\top$ 
and $v$ is the denominator of~$K$. Then
\[
\sum_{i=0}^\rho \ell_i r_i = \frac{\sum_{i=0}^\rho \ell_ia_i}D+\frac{\sum_{i=0}^\rho\ell_i p_i}v.
\]
Notice that $\deg_y(a_i) < \deg_y(D)$ for all $i = 0,\dots, \rho$. So
$\deg_y(\sum_{i=0}^\rho \ell_ia_i) < \deg_y(D)$. Since $D$ is $\sigma_y$-normal
and strongly coprime with~$K$, we see that $\sum_{i=0}^\rho\ell_ia_i/D$ is a 
proper rational function in~$\set F(y)$ whose denominator is $\sigma_y$-normal and strongly
coprime with~$K$. Because $\sum_{i=0}^\rho \ell_ip_i \in \im(\phi_K)^\top$,
so by definition, $\sum_{i=0}^\rho\ell_ir_i$ is again a $\sigma_y$-remainder with respect to~$K$.
By Theorem~\ref{THM:shellred}, along with \eqref{EQ:telescoperred}, we see that 
$L$ is a telescoper for $T$ if and only if $\sum_{i=0}^\rho \ell_i r_i = 0$.

\smallskip\noindent
(ii) Let $\rho\in \set N$ and for $i =0,\dots,\rho$, write $r_i$ in the form \eqref{EQ:ithrem}.
Since $D$ is strongly coprime with~$K$, we have $\gcd(D,v) = 1$. 
Observe that the total number of nonzero terms in each $p_i$
is no more than $\dim(\im(\phi_K)^\top)$ as $p_i\in \im(\phi_K)^\top$.
Therefore, $r_0,\dots,r_\rho$ are linearly dependent over~$\set F$ whenever 
$\rho \geq \deg_y(D) + \dim(\im(\phi_K)^\top)$. By part (i), $T$ has telescopers 
and their minimal order is no more than $\deg_y(D) + \dim(\im(\phi_K)^\top)$,
whose precise formulas follow by
\cite[Proposition~3.12]{CDGHL2025} and the observation that
\[
\deg_y(D) = \deg_y(D_0)
= \begin{cases}
\sum_{j=1}^mk_j\mu_j \deg_z(P_j(z))&\text{in the usual shift case,}\\[.5cm]
\sum_{j=1}^mk_j\mu_j^2\deg_z(P_j(z))&\text{in the $q$-shift case.}
\end{cases}
\]

\smallskip\noindent
(iii) Since $r\neq 0$, the term $T$ is not $\sigma_y$-summable and thus its 
telescopers have orders at least one.
Let $L = \sum_{i=0}^\rho \ell_iS_x^i$ be a telescoper of minimal order for $T$, 
where $\rho \in \set Z^+$ and $\ell_0,\dots,\ell_\rho\in \set F$, not all zero.
Then $\sum_{i=0}^\rho \ell_i r_i = 0$ by part (i).
For $i = 0,\dots,\rho$, write $r_i = h_i + p_i/v$, where $h_i\in\set F(y)$ is proper with denominator $d_i$,
$p_i\in \im(\phi_K)^\top$ and $v$ is the denominator of~$K$. 
Since the $d_i$ are strongly coprime with~$K$, they are all coprime with~$v$.
It follows from $\sum_{i=0}^\rho \ell_i r_i = 0$ that
\[
\ell_0 h_0 + \ell_1 h_1 + \cdots + \ell_\rho h_\rho = 0.
\]
By partial fraction decomposition, for any irreducible factor $a\in\set F[y]$ 
of~$d_0$ over~$\set F$ with multiplicity $k>0$,
there exists an integer $n$ with $1\leq n\leq \rho$ such that  $a^k$ is also a factor of~$d_n$
and then, since $d_n$ is $\sigma_y$-related to $\sigma_x^n(d)$ by Proposition~\ref{PROP:shiftrem},
$a$ is $\sigma_y$-equivalent to an irreducible factor $\tilde a\in\set F[y]$ of $\sigma_x^n(d)$ 
over~$\set F$ with multiplicity at least $k$.
Notice that $d_0$ and $d$ are associates and $d$ admits the factorization \eqref{EQ:ildecomp}, 
where the $p_i$ are $\sigma_y$-inequivalent. So for all $i =1,\dots,m$ and $j\in \Lambda_i$, 
there exist $n_{ij}\in\set N$ with $1\leq n_{ij}\leq \rho$ and $\ell\in\Lambda_i$
such that
\[
k_{ij}\leq k_{i\ell}\quad\text{and}\quad
p_i^{\delta_{\lambda_i,\mu_i}^j}\ \text{is $\sigma_y$-equivalent to}\
p_i^{\delta_{\lambda_i,\mu_i}^{\ell+\lambda_i n_{ij}}},
\ \text{or equivalently,}\ j\equiv \ell + \lambda_i n_{ij} \mod \mu_i.
\]
Let the $n_{ij}$ be the minimal ones with such a property. The assertion follows by the fact that
the term $T$ does not have a telescoper of order less than $\max_{1\leq i\leq m,\, j\in \Lambda_i}\{n_{ij}\}$.
\end{proof}

\begin{example}\label{EX:bounds}
Assume that $\sigma_x$ and $\sigma_y$ are both the usual shift operators. 
Let $T(x,y) = {x+2y\choose y}$. Then $T$ is a $(\sigma_x,\sigma_y)$-hypergeometric 
term whose $\sigma_y$-quotient has a $\sigma_y$-standard kernel 
$K = (x + 2y + 1)(x + 2y + 2)/((y + 1)(x + y + 1))$ and a corresponding shell $S = 1$. 
Moreover, $r=(-x^2+x+2y+2)/(3v)$ is a $\sigma_y$-remainder of~$S$ with respect to~$K$, 
where $v = (y+1)(x+y+1)$. By Theorem~\ref{THM:bounds}, the minimal order of 
telescopers for~$T$ is no more than $\dim(\im(\phi_K)^\top) = 2$, which is 
exactly the real minimal order of telescopers for~$T$.
\end{example}
\begin{example}\label{EX:qbounds}
Assume that $\sigma_x$ and $\sigma_y$ are both the $q$-shift operators. 
Let $T(n,k) = \qbinom{n}{k}_q$.
Then $T$ is a $(\sigma_x,\sigma_y)$-hypergeometric term with 
$\sigma_x(T(n,k)) = T(n+1,k)$, $\sigma_y(T(n,k)) = T(n,k+1)$ and $q^n = x$, 
$q^k = y$, whose $\sigma_y$-quotient has a $\sigma_y$-standard kernel 
$K = (x-y)/(y(qy-1))$ and a corresponding shell $S = 1$. Moreover, $r=(x-y)/v$ is a 
$\sigma_y$-remainder of~$S$ with respect to~$K$, where $v = y(qy-1)$. 
By Theorem~\ref{THM:bounds}, the minimal order of telescopers for~$T$ is 
no more than $\dim(\im(\phi_K)^\top) = 2$, which is exactly the real 
minimal order of telescopers for~$T$.
\end{example}
\section{Comparison of order bounds}\label{SEC:comp}
Order bounds on the telescopers have been well studied in \cite{MoZe2005} where upper 
ones for both the usual shift and the $q$-shift cases were derived, and in \cite{AbLe2005} 
where a lower one for the usual shift case was obtained. 
We are not aware of any other work concerning lower order bounds in the $q$-shift case.
Since a comparison of order bounds has already been provided in \cite{Huan2016} for the usual shift case, 
we restrict our attention in this section to the $q$-shift case and compare our upper order bound
to the one obtained by Apagodu and Zeilberger~\cite{MoZe2005}.

Throughout this section, we let $n,k$ be two distinct discrete variables.
Then $q^n,q^k$ are transcendental over~$\set K$, since $q\in\set K$ is neither zero nor a root of unity. 
This means that all rational functions in $q^n,q^k$ over~$\set K$ constitute a well-defined field, 
denoted by~$\set K(q^n,q^k)$.
By taking $x = q^n$ and $y = q^k$ in the previous sections, everything remains valid.
We now consider a {\em $q$-proper hypergeometric term} over~$\set K(q^n,q^k)$ of the form
\begin{equation}\label{EQ:qproperform}
T(n,k) = p\,\xi^k q^{\gamma k (k-1)/2}
\prod_{i=1}^m\frac{(q;q)_{\alpha_in+\alpha_i'k+\alpha_i''}(q;q)_{\beta_in-\beta_i'k+\beta_i''}}
{(q;q)_{\mu_in+\mu_i'k+\mu_i''}(q;q)_{\nu_in-\nu_i'k+\nu_i''}},
\end{equation}
where $p\in\set K[q^n,q^{-n},q^k,q^{-k}]$, $\xi\in\set K\setminus\{0\}$ is not an integral power of~$q$, 
$\gamma,\alpha_i'',\beta_i'',\mu_i'',\nu_i''\in \set Z$, $m,\alpha_i,\alpha_i',\beta_i,\beta_i',\mu_i,\mu_i',\nu_i,\nu_i'\in \set N$.
Note that the hidden assumption that the number of $q$-Pochhammer symbols is the same for all four kinds 
does not lose any generality, since we can always put further terms $(q;q)_{0n\pm 0k+0}$ to achieve the balance.
Further assume that there exist no integers $i,j$ with $1\leq i,j\leq m$ such that
\[
\begin{array}{cccc}
&\{\alpha_i=\mu_j, &\alpha_i'=\mu_j', & \alpha_i''-\mu_j''\in \set N\}\\[1.5ex]
\text{or} & \{\beta_i = \nu_j, & \beta_i'=\nu_j', & \beta_i''-\nu_j''\in\set N\},
\end{array}
\]
which guarantees that the Laurent polynomial $p$ in \eqref{EQ:qproperform} is as ``large" as possible.
Then Apagodu and Zeilberger stated in \cite[$q$-Theorem]{MoZe2005} 
that the minimal order of telescopers for~$T$ is bounded by
\begin{equation}\label{EQ:AZbound}
B_{AZ} = \max\Big\{\gamma + \sum_{i=1}^m \alpha_i'^2,\sum_{i=1}^m \mu_i'^2\Big\}
+\max\Big\{-\gamma + \sum_{i=1}^m\nu_i'^2,\sum_{i=1}^m\beta_i'^2\Big\},
\end{equation}
which is generically sharp.
We next show that the number $B_{AZ}$ given above is at least the upper bound obtained 
by applying Theorem~\ref{THM:bounds} (ii) to~$T$. 
Reordering the terms in \eqref{EQ:qproperform} if necessary, let $\Lambda$ be the maximum set 
of indices $i\in\set N$ with $1\leq i\leq m$ and satisfying
\[
\begin{array}{cccc}
&\{\alpha_i=\mu_i, &\alpha_i'=\mu_i', & \mu_i''-\alpha_i''\in \set N\}\\[1.5ex]
\text{or} & \{\beta_i = \nu_i, & \beta_i'=\nu_i', &\nu_i''-\beta_i''\in\set N\}.
\end{array}
\]
Then $T$ can be rewritten as
\[
T(n,k) = f \xi^k q^{\gamma k (k-1)/2}
\prod_{i=1,\,i\notin \Lambda}^m
\frac{(q;q)_{\alpha_in+\alpha_i'k+\alpha_i''}(q;q)_{\beta_in-\beta_i'k+\beta_i''}}
{(q;q)_{\mu_in+\mu_i'k+\mu_i''}(q;q)_{\nu_in-\nu_i'k+\nu_i''}},
\]
where $f\in \set K(q^n,q^k)$. By a direct calculation, we confirm that $T(n,k+1)/T(n,k)$ admits 
an RNF of the form
\[
\Big(\xi q^{\gamma k }\prod_{i=1,\, i\notin \Lambda}^m
\frac{(q^{\alpha_i n + \alpha_i'k+\alpha_i''+1};q)_{\alpha_i'} (q^{\nu_i n-\nu_i' k + \nu_i''-\nu_i'+1};q)_{\nu_i'}}
{(q^{\mu_in+\mu_i'k+\mu_i''+1};q)_{\mu_i'}(q^{\beta_i n-\beta_i'k+\beta_i''-\beta_i'+1};q)_{\beta_i'}},
\, f\Big),
\]
where $(a;q)_m = \prod_{i=0}^{m-1}(1-aq^{i})$ for $a\in \set K(q^n,q^k)$ and $m\in \set N$ with $(a;q)_0=1$.
Notice that, by Lemma~\ref{LEM:twornfs}, the numerators, as well as the denominators,
of all kernels of a nonzero rational function in $\set K(q^n,q^k)$ have the same degree in~$q^k$.
So, with the notation introduced in Theorem~\ref{THM:bounds}, we obtain
\begin{align*}
\deg_{q^k}(u) &= \sum_{i=1,\, i\notin \Lambda}^m(\alpha_i'^2+\nu_i'^2) 
+\max\Big\{\sum_{i=1,\, i\notin \Lambda}^m(\beta_i'^2-\nu_i'^2)+\gamma,0\Big\}
\\[1ex]
\text{and}\quad 
\deg_{q^k}(v) &= \sum_{i=1,\, i\notin \Lambda}^m(\mu_i'^2+\beta_i'^2) 
+\max\Big\{-\sum_{i=1,\, i\notin \Lambda}^m(\beta_i'^2-\nu_i'^2)-\gamma,0\Big\}.
\end{align*}
Moreover, observe that $\xi$ is not an integral power of~$q$ and 
the denominator of the shell $f$ consists of integer-linear irreducible factors
dividing either $(1-q^{\mu_i n+\mu_i'k+\ell})$ or $q^{\nu_i'k}(1-q^{\nu_i n-\nu_i'k+\ell})$
for some $i\in \Lambda$ and $\ell\in \set Z$. 
Thus we conclude that the bound given in Theorem~\ref{THM:bounds}~(ii) is at most
\[
\max\{\deg_{q^k}(u),\deg_{q^k}(v)\} 
+ \sum_{i\in\Lambda}\Big(\frac{\mu_i'^2}{\gcd(\mu_i,\mu_i')}+\frac{\nu_i'^2}{\gcd(\nu_i,\nu_i')}\Big),
\]
which is no more than $B_{AZ}$ since $\alpha_i'=\mu_i'$ and $\beta_i'=\nu_i'$ for all $i\in \Lambda$.
In the generic situation, our upper bound in Theorem~\ref{THM:bounds}~(ii)
is as good as~$B_{AZ}$. However, our bound could be much tighter in some cases.

\begin{example}
Consider a $q$-proper hypergeometric term of the form
\[
T = \frac{2^{k+1}}{1-q^{n+\alpha k+\alpha}}-\frac{2^k}{1-q^{n+\alpha k}}+\frac{2^k}{(1-q^{n+\beta k})(1-q^{k})},
\]
where $\alpha,\beta$ are positive integers. We remark that the decomposed form given above is for readability only.
Rewriting $T$ into the form \eqref{EQ:qproperform} yields that
\[
T = p\, 2^k\frac{(q;q)_{n+\alpha k-1}(q;q)_{n+\beta k-1}(q;q)_{k-1}}{(q;q)_{n+\alpha k+\alpha} (q;q)_{n+\beta k} (q;q)_{k}}
\quad\text{for some $p\in\set Q(q)[q^n,q^k]$}.
\]
By~\eqref{EQ:AZbound}, we have $B_{AZ}=\alpha^2+\beta^2+1$.
However, our upper bound given by Theorem~\ref{THM:bounds} (ii) is equal to $\beta^2 + 1$,
which is exactly the minimal order of telescopers for~$T$.
\end{example}

\begin{remark}
Along with \cite[Theorem~4.6]{CHM2005}, the bound described in \cite[$q$-Theorem]{MoZe2005}
can be also applied to non-$q$-proper hypergeometric terms. On the other hand, 
Theorem~\ref{THM:bounds} is already applicable to any $q$-hypergeometric term
provided that its telescopers exist.
\end{remark}

\section*{Acknowledgments}
I would like to express my gratitude to Shaoshi Chen for constructive suggestions 
and support. I also would like to thank Hao Du, Yiman Gao and anonymous referees
for valuable 
comments on the original draft of this paper.
This work was partially supported by the NSFC grant (No.\ 12101105) and the
Natural Science Foundation of Fujian Province of China (No.\ 2024J01271).

\bibliographystyle{plain}
\bibliography{scrfs}

\end{document}